\documentclass[english,review,12pt,a4paper]{article}

\usepackage[utf8]{inputenc}
\usepackage[T1]{fontenc}
\usepackage{babel}

\usepackage{amsthm}
\newtheorem{theorem}{Theorem}
\newtheorem{lemma}{Lemma}

\newtheorem{ex}{Example}

\newenvironment{example-cont}[1]{\bigskip\noindent\textbf{Example~\ref{#1}.~(cont.)\hspace{\labelsep}}}{\bigskip\noindent}

\usepackage{macros}
\usepackage[noblocks]{authblk}

\usepackage[stable,multiple]{footmisc}
\usepackage{manyfoot}
\DeclareNewFootnote{email}[alph]
\DeclareNewFootnote{support}[fnsymbol]
\stepcounter{footnotesupport}

\title{%
	Converting Nondeterministic Two-Way Automata into Small Deterministic Linear-Time Machines%
	\footnote{%
		This work contains, in an extended form, some material and results
		which were previously presented in a preliminary form
		in conference papers~\cite{Pru14} and~\cite{GPPP18}.%
	}%
}

\author[,1]{Bruno Guillon%
	\protect\footnoteemail{bruno.guillon@uca.fr}%
}

\author[,2]{Giovanni Pighizzini%
	\protect\footnoteemail{pighizzini@di.unimi.it}%
	\protect\footnotesupport{Partially supported by Gruppo Nazionale per il Calcolo Scientifico (GNCS-INdAM).}%
}
\author[,2]{Luca Prigioniero%
	\protect\footnoteemail{prigioniero@di.unimi.it}%
}

\author[,3]{Daniel Pr\r u\v sa%
	\protect\footnoteemail{prusapa1@cmp.felk.cvut.cz}%
	\protect\footnotesupport{Supported by the Czech Science Foundation, grant 19-21198S.}%
}

\affil[1]{LIMOS, Université Clermont-Auvergne, France}
\affil[2]{Dipartimento di Informatica, Universit\`a degli Studi di Milano, Italy}
\affil[3]{Faculty of Electrical Engineering, Czech Technical University, Prague}

\date{}

\begin{document}
\maketitle
\begin{abstract}
	\noindent
	In 1978 Sakoda and Sipser raised the question
	of the cost, in terms of size of representations,
	of the transformation of two-way and one-way nondeterministic automata 
	into equivalent two-way deterministic automata. 
	Despite all the attempts,
	the question has been answered only for particular cases
	(\eg, restrictions of the class of simulated automata or of the class of simulating automata).
	However the problem remains open in the general case,
	the best-known upper bound being exponential.
	We present a new approach
	in which unrestricted nondeterministic finite automata 
	are simulated by deterministic models extending two-way deterministic finite automata, 
	paying a polynomial increase of size only.
	Indeed, we study the costs of the conversions
	of nondeterministic finite automata
	into some variants of one-tape deterministic Turing machines working in linear time,
	namely Hennie machines, weight-reducing Turing machines, and weight-reducing Hennie machines.
	All these variants are known to share the same computational power:
	they characterize the class of regular languages.
\end{abstract}


\section{Introduction}
\label{sec:intro}

\emph{One-way deterministic finite automata} (\owdfa\s)
are the canonical acceptor for the class of regular languages.
By allowing nondeterministic transitions (\ownfa\s)
or/and movements of the head in both directions on the input tape,
so obtaining \emph{two-way deterministic} and \emph{nondeterministic finite automata} (\twdfa\s/\twnfa\s),
the computational power does not increase~\cite{RS59,She59}.
Other extensions of finite automata
have been proved to capture the same class of languages,
such as \emph{constant-height pushdown automata}~\cite{GMP10,GPP18},
\emph{straight-line programs}~\cite{GMP10},
\emph{$1$-limited automata}~\cite{WW86,PP14,PP19},
or, as will be of interest for this work,
\emph{linear-time one-tape Turing machines}~\cite{Hen65,Pru14,GPPP18,GPPP21a}.%
\footnote{%
	Actually, the model considered by Hennie was deterministic.
	Several extensions of this result,
	including that to the nondeterministic case
	and greater time lower bounds for nonregular language recognition,
	have been stated in the literature~\cite{Tr64,Har68,Mic91,Pig09,TYL10}.%
}

A natural question concerning models
that share the same computational power
is the comparison of the sizes of their descriptions.
In particular, the cost of the elimination of nondeterminism is a standard problem.
For instance, it is a classical result
that an exponential increase in size
is sufficient and, in the worst case, necessary for the conversion of \ownfa\s to \owdfa\s~\cite{RS59}.
However,
already for two-way automata,
the famous Sakoda and Sipser question
concerning the size blowups from \ownfa\s or \twnfa\s to \twdfa\s
is a much more intricate problem.
For both conversions, Sakoda and Sipser conjectured
that the costs are exponential~\cite{Sak78}. 
The question has been solved in some special cases
that can be grouped in three classes:
by considering restrictions on the simulating machines
(\eg, \emph{sweeping}~\cite{Sip80}, \emph{oblivious}~\cite{HS03}, \emph{few reversals} \twdfa\s~\cite{Kap13}),
by considering restrictions on languages
(\eg, \emph{unary case}~\cite{GMP03}),
by considering restrictions on the simulated machines
(\eg, \emph{outer-nondeterministic} automata~\cite{GGP14,KP15}).
However, in spite of all attempts,
in the general case the question remains open
(for further references see~\cite{Pig13}).
Here, we consider a different approach:
in order to obtain a polynomial simulation,
we enlarge the family of simulating machines.
To this end, we study size blowups for the conversion of \ownfa\s and \twnfa\s
into several variants of \emph{linear-time one-tape deterministic Turing machines},
which all characterize regular languages.
These variants and their properties have been investigated in~\cite{GPPP18,GPPP21a}.
We now give a short description of them (see \cref{fig:machines}).
\begin{figure}[tb]
	\centering
	\includegraphics[width=\textwidth]{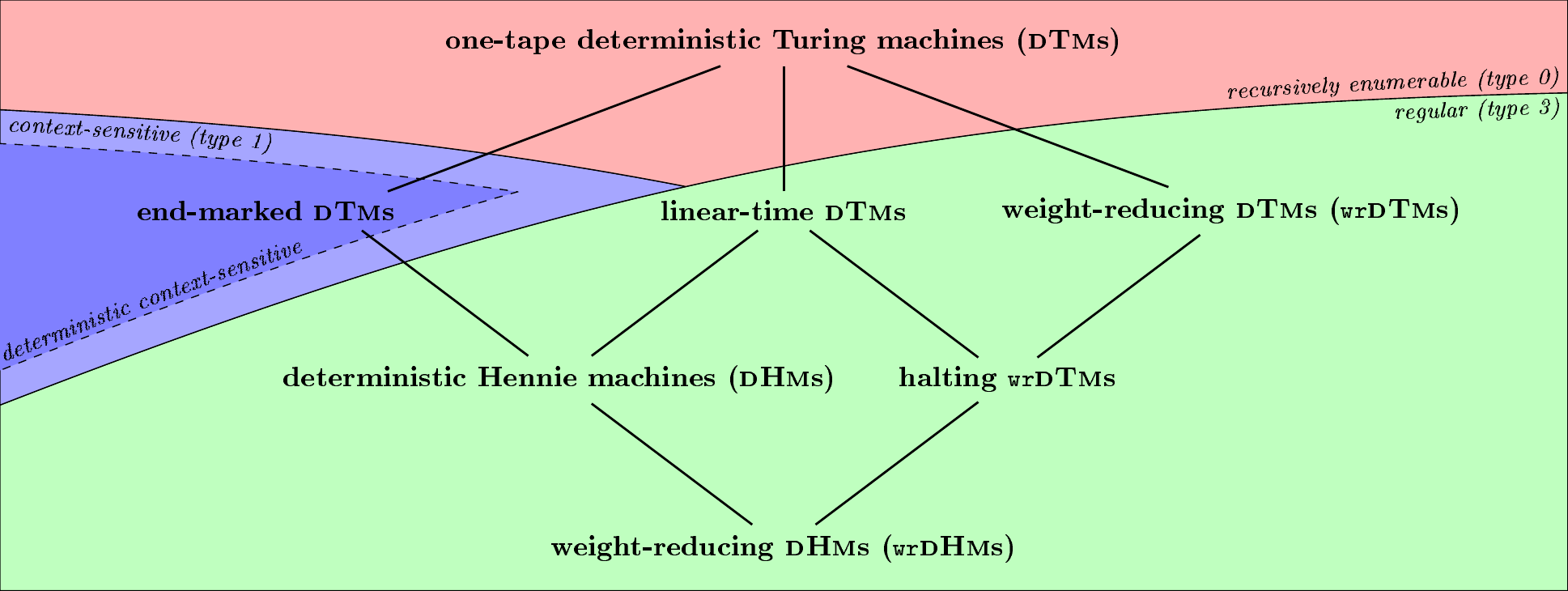}
	\caption{%
		Variants of one-tape deterministic Turing machines
		and their expressive power confronted with the Chomsky hierarchy.
		In particular,
		end-marked \dtm\s are known as \emph{deterministic linear bounded automata} in the literature,
		and recognize the so-called \emph{deterministic context-sensitive languages},
		a subclass of context-sensitive languages,
		see~\cite{Wal70}.
		It is still unknown if such inclusion is strict.
	}
	\label{fig:machines}
\end{figure}
\medbreak

As proven by Hennie,
linear-time one-tape deterministic Turing machines
recognize regular languages only~\cite{Hen65}.
However,
it cannot be decided whether or not a one-tape Turing machine actually works in linear time.
This negative result remains
true in the restricted case of \emph{end-marked machines}, 
namely one-tape deterministic Turing machines that do not have any extra space,
besides the tape portion which initially contains the input.
End-marked machines working in linear time will be called deterministic \emph{Hennie machines}.
To overcome the above-mentioned ``negative'' results, 
a syntactical restriction on deterministic one-tape Turing machines,
called \emph{weight-reducing machines},
has been considered.
This restriction enforces computations
to either be infinite
or to halt within a linear number of steps in the input length. 
In contrast to Hennie machines,
this model benefits from nice properties.
In particular, it can be decided whether a one-tape Turing machine is weight-reducing,
and whether a weight-reducing Turing machine is halting whence works in linear time.
Furthermore, haltingness of the model can always be obtained
paying a polynomial size increase only.
However, its space is not limited to the portion of the tape
that contains the input at the beginning of the computation,
namely the device is not end-marked.
Considering end-marked linear-time machines
satisfying the syntactical restriction of weight-reducing machines,
we obtain \emph{weight-reducing Hennie machines}. 
\medbreak

Our main result is
that each \twnfa~\mA can be simulated
by a one-tape deterministic Turing machine
which works in linear time
(with respect to the input length)
and which has polynomial size
with respect to the size of~\mA.
We point out that the resulting machine
can use extra space,
besides the tape segment which initially contains the input.
Next, the machine is halting and \emph{weight-reducing}, thus implying
a linear execution time.
Hence, nondeterminism can be eliminated
with at most a polynomial size increase,
obtaining a linear execution time in the input length,
and provided the ability to rewrite tape cells
and to use some extra space.

We then investigate
what happens when the latter capability is removed,
namely if the machine does not have any further tape storage,
\ie, it is a Hennie machine.
We prove that even under this restriction
it is still possible to obtain a machine of polynomial size,
that is each \twnfa can be transformed
into an equivalent Hennie machine of polynomial size.
However,
the machine resulting from our construction
is not weight reducing,
unless we require that it agrees with the given \twnfa
only on sufficiently long inputs.
We do not have this problem in the unary case,
namely for a one-letter input alphabet,
where we prove that each unary \twnfa can be simulated
by a weight-reducing Hennie machine of polynomial size.
Similar results are obtained for the transformation of
\ownfa\s into variants of one-tape deterministic machines.
\medbreak

The paper is organized as follows.
In \cref{sec:prel} we present some fundamental notions and definitions,
included those related to the computational models we are interested in,
and we state some basic properties.
In \cref{sec:2nfa->wrotm} we present our main simulation result:
each $n$-state \twnfa can be transformed into an equivalent halting weight-reducing machine of size polynomial in~$n$.
In \cref{sec:2nfa->dhm} we discuss how the simulation changes
if the resulting machine is required to be a Hennie machine.
Finally, in \cref{sec:1nfa->dhm}
we revise the results of \cref{sec:2nfa->wrotm,sec:2nfa->dhm}
under the assumption that the simulated automata are \emph{one-way} instead of being two-way.

\section{Preliminaries}%
\label{sec:prel}

In this section we recall some basic definitions and notations.
We also describe the main computational models considered in the paper
and we give some preliminary results.

We assume the reader familiar
with notions from formal languages
and automata theory
(see, \eg,~\cite{HU79}).
Given a set~$S$,
\defd{$\card S$} denotes its cardinality
and~\defd{$\powerset S$} the family of all its subsets.
Given an alphabet~\alphab,
\defd{$\length w$} denotes the length of a string~$w\in\alphab^*$,
\defd{$\ith w$} denotes the~$i$\=/th symbol of~$w$, $i=1,\ldots,\length w$,
and~\defd{\emptyword} denotes the empty string.
\smallbreak

\emph{Finite automata}
are computational devices
equipped with a finite control and a finite read-only tape
which is scanned by an input head.
A \defd{one-way nondeterministic finite automaton}
(\defd\ownfa)
is defined as a quintuple~$\mA=\struct{Q,\alphab,\delta,q_0,F}$,
where~$Q$ is a finite \defd{set of states},
\ialphab~is a finite \defd{input alphabet},
$q_0\in Q$ is the \defd{initial state},
$F\subseteq Q$ is a \defd{set of final states},
and~%
$%
\delta:%
Q\times\ialphab%
\rightarrow%
\powerset{Q}%
$
is a \defd{nondeterministic transition function}.
At each step,
according to its current state~$p$
and the symbol~$\sigma$ scanned by the head,%
~\mA enters one nondeterministically-chosen state from~$\delta(p,\sigma)$
and moves the input head rightward to the next symbol.
The machine \defd{accepts} the input
if there exists a computation
starting from the initial state~$q_0$
with the head on the leftmost input symbol
and ending in a final state $q\in F$
after having read the whole input.
%
%
A \ownfa~\mA is said to be \defd{deterministic} (\defd\owdfa),
whenever~$\card{\delta(q,\sigma)}\leq1$,
for any~$q\in Q$ and
$\sigma\in\ialphab$.

Providing \ownfa\s (\resp, \owdfa\s)
with the ability of moving the head back and forth,
we obtain \defd{two-way nondeterministic} (\resp, \defd{deterministic}) \defd{finite automata}
(\defd{\twnfa\s}, \resp, \defd{\twdfa\s}).
They are defined by extending the transition function
so that a left (\lmove) or right (\rmove) head direction is indicated in each instruction.
Furthermore, to prevent the head to fall out the input,
the device is \defd{end-marked}, in the following sense.
Two special symbols~\defd{\lend} and~\defd{\rend}
not belonging to~$\ialphab$,
called the \defd{left} and the \defd{right endmarker}, respectively,
surround each input word,
and enforce the computation to stay between them
(except at the end of computation when accepting, as described below).
More precisely, on input~$w$
the tape contains~$\lend w\rend$,
the left endmarker being at position~$0$
and the right endmarker being at position~$\length w+1$.
By~$\ialphab_{\lend,\rend}$ we denote the set~$\ialphab\cup\set{\lend,\rend}$.
Formally, the transition function of a two-way automaton is
$
\delta:
Q\times\ialphab_{\lend,\rend}%
\rightarrow%
\powerset{Q\times\moveset}%
$
such that,
for each transition~$(q,\xmove)\in\delta(p,\sigma)$,
if~$\sigma=\lend$ then~$\xmove=\rmove$,
and if~$\sigma=\rend$ then~$\xmove=\lmove$ or~$q\in F$.
In this way,
the head cannot violate the endmarkers,
except at the end of computation to accept.
The machine \defd{accepts} the input
if there exists a computation
starting from the initial state~$q_0$
with the head on the $0$-th tape cell
(\ie, scanning the left endmarker)
and ending in a final state $q\in F$
after violating the right endmarker.

The other main computational model we consider
is the \defd{deterministic one-tape Turing machine} (\dtm).
Such a machine
is a tuple~$\struct{Q,\ialphab,\walphab,\delta,q_0,F}$
where~$Q$ is the \defd{set of states},
\ialphab is the \defd{input alphabet},
\walphab is the \defd{working alphabet}
including both \ialphab
and the special \defd{blank symbol}, denoted by~$\tblank$,
that cannot be written by the machine,
$q_0\in Q$ is the \defd{initial state},
$F\subseteq Q$ is the \defd{set of final states},
and~%
$\delta: Q\times\walphab\to Q\times(\walphab\setminus\set\tblank)\times\moveset$
is the partial \defd{deterministic transition function}.
In one step,
depending on its current state~$p$
and on the symbol~$\sigma$ read by the head,
a \dtm changes its state to~$q$,
overwrites the corresponding tape cell with~$\tau$
and moves the head one cell to the left when~$\xmove=\lmove$ or to the right when~$\xmove=\lmove$,
if~$\delta(p,\sigma)=(q,\tau,\xmove)$.
Since~$\delta$ is partial,
it may happen that no transition can be applied.
In this case, we say that the machine \defd{halts}.
At the beginning of computation
the input string~$w$ resides on a segment of a bi-infinite tape,
called \defd{initial segment},
and the remaining infinity of cells contain the blank symbol.
The computation over~$w$
starts in the initial state
with the head scanning the leftmost non-blank symbol.
The input is \defd{accepted}
if the machine eventually halts in a final state.

Let~$\mA=\struct{Q,\ialphab,\walphab,\delta,q_0,F}$ be a \dtm.
A \defd{configuration} of~\mA
is given by the current state,
the tape contents,
and the position of the head.
If the head is scanning a non-blank symbol,
we describe it by~\defd{\config zqu}
where~$zu\in\walphab^*$ is the finite non-blank contents of the tape,
$u\neq\emptyword$,
and the head is scanning the first symbol of~$u$.
Otherwise, we describe it by~\defd{\config{}{q}{\tblank z}} or~\defd{\config zq{}}
according to whether the head is scanning
the first blank symbol to the left or to the right of the non-blank tape contents~$z$, respectively.
If, from a configuration~$\config zqu$
the device may enter in one step a configuration~$\config{z'}{q'}{u'}$,
we say that~$\config{z'}{q'}{u'}$ is \defd{a successor} of~$\config zqu$,
denoted~$\config zqu\yields\config{z'}{q'}{u'}$.
A \defd{halting configuration} is a configuration that has no successor.
The reflexive and transitive closure of~$\yields$
is denoted by~\defd{\yields*}.
On an input string~$w\in\ialphab^*$,
the \defd{initial configuration} is~$\config{}{q_0}{w}$.
An \defd{accepting configuration} is a halting configuration~$\config{z}{q_f}{u}$
such that~$q_f$ is a final state of the machine.
A \defd{computation} is a (possibly infinite) sequence of successive configurations.
It is \defd{accepting} if it is finite,
its first configuration is initial,
and its last configuration is accepting.
Therefore,~%
\[
	\langof\mA=\set{
		w\in\ialphab^*
		\mid
			\config{q_0}{w}{}
			\yields*
			\config z{q_f}u,\,
			\text{where $q_f\in F$ and $\config z{q_f}u$ is halting}
		}%
	\text.
\]
We say that two machines~\mTM and~\mHM \emph{agree} on some language~$L$
if every string in~$L$ is accepted by~\mTM \ifof it is accepted by~\mHM.

The notions of configurations, successors, computations, and halting configurations
naturally transfer to one-way and two-way finite automata.

In the paper we consider the following restrictions of \dtm\s.
\begin{description}
	\item[End-marked machines.]
		We say that a \dtm is \defd{end-marked},
		if at the beginning of the computation the input string is surrounded
		by two special symbols belonging to~$\walphab$,
		\defd{\lend} and~\defd{\rend} respectively,
		called the \emph{left} and the \emph{right endmarkers},
		which can never be overwritten,
		and that prevent the head to fall out the tape portion
		that initially contains the input.
		Formally, for each transition~$\delta(p,\sigma)=(q,\tau,\xmove)$,
		$\sigma=\lend$ (\resp,~$\sigma=\rend$) implies%
		~$\tau=\sigma$ and~$\xmove=\rmove$ (\resp,~$\xmove=\lmove$).
		This is the deterministic restriction
		of the well-known \emph{linear-bounded automata}%
		~\cite{Kur64}.
		For end-marked machines,
		the initial configuration on input~$w$
		is~$\config{}{q_0}{\lend w\rend}$.
	\item[Weight-reducing Turing machines.]
		A \dtm is \defd{weight-reducing} (\defd\wrdtm),
		if there exists a partial order~$<$ on~\walphab
		such that each rewriting is decreasing,
		\ie, $\delta(p,\sigma)=(q,\tau,\xmove)$ implies~$\tau<\sigma$.
		By this condition, in a \wrdtm the number of visits to each tape cell
		is bounded by a constant. However, one \wrdtm could have non-halting
		computations which, hence, necessarily visit infinitely many tape cells.
	\item[Linear-time Turing machines.]
		A \dtm is said to be \defd{linear-time} if
		over each input~$w$,
		its computation halts within~$\bigoof{\length w}$ steps.
	\item[Hennie machines.]
		A \defd{Hennie machine} (\defd\dhm)  
		is a linear-time \dtm which is,
		furthermore, end-marked.
	\item[Weight-Reducing Hennie machines.]
		By combining previous conditions,
		\defd{weight-reducing Hennie machines} (\defd\wrdhm) are defined as particular \dhm,
		for which there exists an order~$<$
		over~$\walphab\setminus\set{\lend,\rend}$
		such that $\delta(p,\sigma)=(q,\tau,\xmove)$ implies~$\tau<\sigma$
		unless $\sigma\in\set{\lend,\rend}$.
		Observe that each end-marked \wrdtm can execute 
		a number of steps which is at most linear in the length of the input.
		Hence,
		end-marked \wrdtm
		are necessarily weight-reducing Hennie machines.
\end{description}

The \defd{size} of a machine
is given by the total number of symbols used to write down its description.
Therefore, the size of a one-tape Turing machine is bounded by a polynomial
in the number of states and of working symbols,
namely, it is~$\thetaof{\card Q\cdot\card\walphab\cdot\log(\card Q\cdot\card\walphab)}$.
In the case of nondeterministic (\resp, deterministic) finite automata,
since no writings are allowed
and hence the working alphabet is not provided,
the size is linear in the number of instructions and states,
which is bounded by a function quadratic (\resp, subquadratic) in the number of states
and linear in the number of input symbols,
namely, it is~$\thetaof{\card\ialphab\cdot\card Q^2}$
(\resp, $\thetaof{\card\ialphab\cdot\card Q\cdot\log(\card Q)}$).
\bigbreak

We now state two preliminary results
that will be used in the subsequent sections
for building weight-reducing Turing machines
and weight-reducing Hennie machines,
respectively.

It is known that a \dtm works in linear time
\ifof there exists a constant~$k$
such that in no computation
the head visits each tape cell more than~$k$ times~\cite{Hen65}.
The following lemma states that,
whenever~$k$ is known,
a weight-reducing Turing machine
can be obtained,
augmenting linearly the working alphabet only.
Indeed, we can enforce the machine
to store on each tape cell
the number of further visits the head is allowed to perform on the cell.
As this number decreases, the overwriting is decreasing.
Using~$k$ copies of each working symbols is enough to implement this.
Therefore, in order to define a \wrdtm,
it is sufficient to define a \dtm
and to provide a constant~$k$
bounding the number of visits of each tape cell.
\begin{lemma}[\cite{GPPP21a}]
	\label{lemma:constantly-many-visits}
	Let $\mTM=\struct{Q,\Sigma,\Gamma,\delta,q_0,F}$ be a 
	\dtm such that, for any input, \mTM performs at most $k$ computation steps on each tape cell.
	Then there is a \wrdtm $\mA$ accepting $L(\mTM)$
	with the same set of states~$Q$ as~\mTM
	and working alphabet of size $\bigoof{k\cdot\card\Gamma}$.
	Furthermore,
	on each input \mA uses the same space as~\mTM.
	Hence,
	if~\mTM is linear time or end-marked
	then so is~\mA.
\end{lemma}

Weight-reducing Turing machines extend weight-reducing Hennie machines
by allowing the use of some extra space besides the portion that initially contains the input.
Indeed, the former model can use a bi-infinite tape while the latter is end-marked.
However, it has been shown that
every finite computation of a \wrtm uses a constant amount of this extra space~\cite{GPPP21a}.
We do not know whether the use of this extra space can be avoided in general,
while keeping the weight-reducing property and bounding the size increase by a polynomial.
Nevertheless, this can be achieved when the inputs are long enough.
\begin{lemma}
	\label{lem:wrTm->wrHm}
	Let~\mTM be a weight-reducing Turing machine
	which uses at most~$C$ initially-blank cells
	in every halting computation.
	Then, there exists a weight-reducing Hennie machine~\mHM
	of size polynomial in the size of~\mTM,
	which agrees with~\mTM
	on every input of length at least~$C$.
\end{lemma}
\begin{proof}
	The idea of the proof
	is the same as the folkloric
	simulation of Turing machines working on a bi-infinite tape
	by Turing machines working on a semi-infinite tape.
	Indeed, we fold the portion of the tape occurring
	to the left of the initial segment
	onto the complementary portion of the tape,
	thus creating a second track.
	Similarly, we can fold the portion of the tape
	occurring to the right of the initial segment
	onto the complementary portion of the tape.
	Next, as~\mTM uses at most~$C$ initially-blank cells in total,
	and providing the input has length at least~$C$,
	we observe that the additional tracks do not overlap
	hence only one additional track is sufficient for the simulation.
	On shorter inputs, the simulation could fail
	giving an outcome different from those of~\mTM.
	Doing so,
	we obtain a \dhm~\mHM
	that uses twice the number of states of \mTM
	(two copies of each state of~\mTM for indicating which track should be read),
	and the working alphabet~$\walphab_{\mHM}=\walphab\cup\walphab^2$
	where~$\walphab$ is the working alphabet of~\mTM.
	In particular, the size of~\mHM is polynomial
	in the size of~\mTM.
	Finally,
	we can extend the order~$<_{\mTM}$ on~$\walphab$
	witnessing that~\mTM is weight-reducing,
	to an order~$<_{\mHM}$ on~$\walphab_{\mHM}$
	witnessing that~\mHM is weight-reducing.
\end{proof}

\section{Simulating Two-way Automata by Weight-reducing Machines}
\label{sec:2nfa->wrotm}

This section is devoted to present our main simulation:
we show that every \twnfa
$\mA=\struct{Q,\ialphab,\delta,q_0,F}$
can be transformed into an equivalent \wrdtm
of size polynomial in the size of~\mA.
Our construction is based on the classical simulation
of \twnfa\s by \owdfa\s,
inspired from Shepherdson's construction~\cite{She59}.
The main idea is to perform forward moves,
while updating a table of size~$n=\card Q^2$
that describes parts of computations
which may occur to the left of the current position.
In parallel,
an adaptation of the classical powerset construction
for converting \ownfa\s into \owdfa\s
is done,
in such a way that the set of states
that are accessible from the initial configuration
when visiting for the first time the current head position
is updated at each move.
In the simulation by \owdfa\s,
the table and the set
are stored on the finite state control.
In our simulation by \wrdtm\s
they will be written,
under a suitable encoding,
in~$\bigoof{n}$ many tape cells.

To describe computation paths
that occur on some restricted part of the tape,
we define \defd{partial configurations},
by relaxing the definition of configurations,
as strings~$\config xqy$
where~$q$ is a state
and ${xy\in\set{\lend,\emptyword}\cdot\ialphab^*\cdot\set{\rend,\emptyword}}$
is a factor of the tape content.
The successor relation~\yields on configurations
extends onto partial configurations.
In particular,
$\config zpX\yields*\config{zX}q{}$ with~$\length X=1$
means that there exists a computation path
\begin{itemize}
\item starting from the rightmost position of~$zX$ (with the head scanning the symbol~$X$) in state~$p$,
\item ending while entering the cell to the right of this position in state~$q$, and
\item which visits only cells
from the part of the tape containing~$zX$
in the meantime.
\end{itemize}
By storing in a set~$\tau_{zX}$ the pairs of states~$(p,q)$
such that such a computation path exists,
we save the possible behaviors of~\mA
when moving backward from the current position.
Indeed, since acceptance is made after violating the right endmarker,
from such a point
the device should eventually turn back forward from the current position,
in order to reach an accepting configuration.
We are going to give the formal definition of the set~\defd{$\tau_{zX}$} for a prefix~$zX$ of the tape content,
together with the definition of the set~\defd{$\gamma_{zX}$}
of states that are reachable from the initial configuration,
when visiting for the first time
the position to the right of the part containing~$zX$.
Formally, for a prefix~$zX\in\set\lend\cdot\ialphab^*\cdot\set{\emptyword,\rend}$
of the tape content with~$\length X=1$:
\begin{align*}
	\tau_{zX}&=\set{(p,q)\in Q\times Q\mid\,\config zpX\yields*\config{zX}q{}\,}
	\text{, and}\\
	\gamma_{zX}&=\set{r\in Q\mid\,\config{}{q_0}{zX}\yields*\config{zX}r{}\,}
	\text.
\end{align*}
Observe that a word~$w\in\ialphab^*$ is accepted by~\mA
\ifof
${F\cap\gamma_{\lend w\rend}\neq\emptyset}$.
In order to simulate~\mA on input~$w$,
it is thus sufficient to incrementally compute~$\gamma_z$
for each prefix~$z$ of~$\lend w\rend$.
To do so,
we will keep updated the table~$\tau_z$ as well.
Indeed, given~$\gamma_z$, $\tau_z$ and a symbol~$\sigma$,
it is possible to compute~$\gamma_{z\sigma}$ and~$\tau_{z\sigma}$.
This is achieved by observing that (see \cref{fig:tau}):
\begin{enumerate}
	\item
		$(p,q)\in\tau_{z\sigma}$
		if and only if
		there exists a sequence~$r_0,s_0,r_1,s_1,\ldots,r_\ell\in Q$,
		with~$\ell\geq0$, satisfying:
		\begin{itemize}
		\item $r_0=p$,
		\item $(q,\rmove)\in\delta(r_\ell,\sigma)$, and
		\item $(s_i,\lmove)\in\delta(r_i,\sigma)$
		and $(s_i,r_{i+1})\in\tau_z$, for~$i=0,\ldots,\ell-1$.
		\end{itemize}
	\item
		$q\in\gamma_{z\sigma}$
		\ifof
		there exists~$p\in\gamma_z$
		such that~$(p,q)\in\tau_{z\sigma}$.
\end{enumerate}
\begin{figure}[tb]
	\centering
	\tikzset{%
	tape vlines/.append style={very thin,dotted},
	state/.append style={minimum size=1.25em,inner sep=0,font={\small},draw=none}
}

\begin{tikzpicture}
	\tape[0][prefix]{}[\lend]{3}[0][tape cell symb][0][2]
	\tape[0][suffix]{(prefix lastcell)++(0:3*\cellwidth)}{1}
	\tape[0][window]{(suffix lastcell)++(0:\cellwidth)}[$X$,$\sigma$]{3}[2.5][tape cell symb][0][2]

	\path
		(prefix cell 1.north west)
		edge[decorate,decoration={brace,raise=2*\cellheight,amplitude=12*\cellheight}]
		node[above,yshift=12*\cellheight]	{$z$}
		(window cell 1.north east)
		;
	\path
		(window cell 2.south)	++(270:\cellheight/2)	node[state]	(run r0) {$r_0$}
		(run r0)			++(270:3pt)		++(0:-\cellwidth)	node[state]	(run s0)	{$s_0$}
		(run r0.south)	++(270:.1*6.5)							node[state]	(run r1)	{$r_1$}
		(run r1)			++(270:3pt)		++(0:-\cellwidth)	node[state]	(run s1)	{$s_1$}
		(run r1.south)	++(270:.1*8.5)							node[state]	(run rl)	{$r_\ell$}
		(run rl)			++(270:2pt)		++(0:\cellwidth)	node[state]	(run q)	{$q$}
		($(run r1)!.5!(run rl)$)	node[rotate=90]	{$\cdots$}
		;

	\draw[->,>=stealth]	(run r0)	-- (run s0);
	\drawsmoothrun[-.075]{(run s0.west)}{-3,-1,1,-1,-1,1,4}{->,>=stealth}{(run r1.160)}
	\draw[->,>=stealth]	(run r1)	-- (run s1);
	\begin{scope}
		\drawsmoothrun[-.075]{(run s1.west)}{-3,-1,1,1,-1,-2,-1,1,4,1,-1,1}{->,>=stealth,draw=none}{(run rl.west)}
	\end{scope}
	\begin{scope}
		\clip
			(run s1.north) -- ++(180:6.5*\cellwidth)
			-- ($(control-8)+(180:1.225*\cellwidth)+(270:\pgflinewidth)$)
			-- ($(control-8)+(270:\pgflinewidth)$)
			-- ($(control-8)+(90:\pgflinewidth)$)
			-- ++(0:5.25*\cellwidth)
			-- cycle;

		\drawsmoothrun[-.075]{(run s1.west)}{-3,-1,1,1,-1,-2,-1,1,4,1,-1,1}{->,>=stealth}{(run rl.west)}
	\end{scope}
	\begin{scope}
		\clip
			($(control-8)+(270:\pgflinewidth)$)
			-- ($(control-8)+(90:\pgflinewidth)$)
			-- ($(control-9)+(90:\pgflinewidth)$)
			-- ($(control-9)+(270:\pgflinewidth)$)
			--cycle
			;
		\drawsmoothrun[-.075]{(run s1.west)}{-3,-1,1,1,-1,-2,-1,1,4,1,-1,1}{->,>=stealth,dashed}{(run rl.west)}
	\end{scope}
	\begin{scope}
		\clip
			($(control-9)+(270:\pgflinewidth)$)
			-- ($(control-9)+(90:\pgflinewidth)$)
			-- ++(0:2.25*\cellwidth)
			-- (run rl.south)
			-- ++(180:7.5*\cellwidth)
			--cycle
			;
		\drawsmoothrun[-.075]{(run s1.west)}{-3,-1,1,1,-1,-2,-1,1,4,1,-1,1}{->,>=stealth}{(run rl.west)}
	\end{scope}
	\draw[->,>=stealth] (run rl)	--	(run q);
\end{tikzpicture}
	\caption{%
		A computation path from~$p=r_0$ to~$q$
		giving~$(p,q)\in\tau_{z\sigma}$.
		For each~$i$, $(s_i,r_{i+1})\in\tau_z$
		and $(s_i,\lmove)\in\delta(r_i,\sigma)$,
		while~$(q,\rmove)\in\delta(r_\ell,\sigma)$.
	}
	\label{fig:tau}
\end{figure}
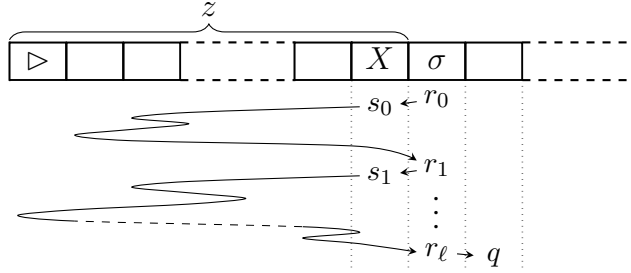
We represent a pair~$(\gamma_z,\tau_z)$
as a word~$uv$ in~$\set{0,1}^*$
with~${\length u=n}$ and~${\length v=n^2}$.
Each bit of~$u$ (\resp,~$v$)
indicates the membership of some state~$p$
(\resp, some pair~$(p,q)$ of states)
to the set~$\gamma_z$ (\resp,~$\tau_z$)
through an implicitly fixed bijection
from~$Q$ to~$\set{1,\ldots,n}$
(\resp, from~$Q^2$ to~$\set{1,\ldots,n^2}$).
For each input symbol~$\sigma$,
there exists a halting \dhm~\mTM[\sigma]
which computes~$(\gamma_{z\sigma}, \tau_{z\sigma})$
from~$(\gamma_z, \tau_z)$
in the following sense.
On input~$uv\in\set{0,1}^{n+n^2}$ encoding~$(\gamma_z,\tau_z)$,
\mTM[\sigma] ends the computation
with the tape containing the encoding~$u'v'\in\set{0,1}^{n+n^2}$ of~$(\gamma_{z\sigma},\tau_{z\sigma})$.
Notice that this computation does not depend on the entire~$z$,
which, indeed, is not given to~\mTM[\sigma],
but only on the information on~$z$ stored in~$\gamma_z$ and~$\tau_z$ which are given in input.

\begin{lemma}
	\label{lem:update-tables}
	For each~$\sigma\in\ialphab\cup\set\rend$,
	there exists a halting \dhm~\mTM[\sigma]
	with~$\bigoof{n^6}$ states
	and~$\bigoof{1}$ working symbols
	that on input~$(\gamma_z,\tau_z)$
	halts with the tape containing~$(\gamma_{z\sigma},\tau_{z\sigma})$
	after~$\bigoof{n^5}$ visits to each cell.
	The input and the output are represented on the tape as strings in~$\set{0,1}^{n+n^2}$.
\end{lemma}
\begin{proof}
	Fixed~$\sigma$, let~$uv$ be the input string encoding the pair of tables~$(\gamma_z,\tau_z)$,
	of size~$n$ and~$n^2$ respectively.
	In order to update them,%
	~\mTM[\sigma] uses a second track on the tape,
	on which it will progressively build the updated tables.
	At the end of the computation,
	namely when the updated tables have been determined
	and written down over the second track,
	the device performs a projection of the tape on its second track,
	in order to produce the correct output,
	and halts.
	
	We fix the working alphabet~$\walphab=\set{0,1}\,\cup\,\set{0,1}^2$.
	The ``simple'' symbols from~$\set{0,1}$
	are used only for the input and the output of~\mTM[\sigma].
	From now on,
	we suppose that the tape contains only symbols
	from the $2$-track alphabet part,
	\ie, the right side of the union.
	Moreover,
	since the length of the input is fixed and it is~$n+n^2$,
	we can suppose that~\mTM[\sigma]
	keeps updated a state component of size~$n+n^2$
	which always stores the position of its head on the tape.
	This allows it to navigate over the tables.

	We divide the tape into two parts:
	a prefix~$\overline u$ of length~$n$
	(thus covering the factor~$u$ which encodes~$\gamma_z$ on its first track)
	and a suffix~$\overline v$ of length~$n^2$
	(thus covering the factor~$v$ which encodes~$\tau_z$ on its first track).
	As previously explained,
	the updated table~$\gamma_{z\sigma}$
	can easily be obtained
	once the updated table~$\tau_{z\sigma}$ has been computed.
	Hence, we first show how to write the table~$\tau_{z\sigma}$
	on the second track of~$\overline v$.
	This is achieved using the space~$n$ available
	on the second track of~$\overline u$
	as temporary memory,
	to which we refer
	as \emph{temporary table}.
	\smallbreak

	As observed above (see~\cref{fig:tau}),
	a computation path
	on the segment containing~$z\sigma$
	starting from the rightmost position of the segment
	and exiting the segment to the right at its last step,
	\ie, a computation of the form~$\config zp\sigma\yields*\config{z\sigma}q{}$,
	can be decomposed into an alternation
	of computation paths
	on the segment containing~$z$
	(described by the table~$\tau_z$)
	and of backward computation steps on~$\sigma$ connecting these paths,
	followed by a last forward computation step on~$\sigma$
	that exits the segment.
	For each state~$p$,
	in order to decide which pairs~$(p,q)$ belong to~$\tau_{z\sigma}$,
	the machine~\mTM[\sigma] first computes
	the set~\defd{$Z_p$} of states
	that are reachable at the rightmost position
	of the segment containing~$z\sigma$,
	from the state~$p$ at the same position,
	by visiting only cells from the segment, \ie,
	\begin{equation*}
		Z_p=\set{r\mid \config zp\sigma\yields*\config zr\sigma}
		\text.
	\end{equation*}
	Thus, a pair~$(p,q)$ belongs to~$\tau_{z\sigma}$
	\ifof for some~$r\in Z_p$,
	we have ${(q,\rmove)\in\delta(r,\sigma)}$.
	For a fixed~$p$,%
	~\mTM[\sigma] can incrementally construct~$Z_p$
	on the temporary table as follows.
	Initially,
	all the cells from the table
	are unmarked (\ie, contain~$0$)
	except the one corresponding to state~$p$
	which contains~$1$.
	The update process behaves as follows:
	for each state~$r$ corresponding to a marked cell,
	each state~$s$ such that~$(s,\lmove)\in\delta(r,\sigma)$,
	and each state~$r'$ such that~$(s,r')\in\tau_z$,
	the machine marks the cell corresponding to~$r'$ with~$1$
	in the temporary table.
	Since~$Z_p$ has size bounded by~$n$,
	after at most~$n$ passes,
	the temporary table is not modified anymore
	and contains exactly an encoding of~$Z_p$.

	So done, computing the set~$Z_p$
	uses only a polynomial number of states in~$n$.
	This is however not sufficient
	to get a weight-reducing machine of polynomial size.
	To this end, using \cref{lemma:constantly-many-visits},
	we should indeed prove that
	the number of visits to each cell
	is bounded by some polynomial in~$n$.
	To update~$Z_p$,
	three nested loops on states,
	namely on~$r$, $s$, and~$r'$,
	are used.
	Once such a triple is fixed,
	the machine navigates on the tape
	in order to check
	that~%
	$r$ is currently marked in the temporary table,
	and~%
	$(s,r')\in\tau_z$
	(%
		notice that the condition~$(s,\lmove)\in\delta(r,\sigma)$
		is verified in constant time,
		since~$\sigma$ is fixed%
	).
	These two conditions require to read the corresponding cells
	in the temporary table (on~$\overline u$)
	and in the table~$\tau_z$ (on~$\overline v$),
	respectively.
	This can be performed by visiting each tape cell at most twice.
	Marking the cell corresponding to~$r'$
	also implies to scan the tape part~$\overline u$ twice.
	As the operation is repeated for each triple,
	we obtain that the number of visits to each cell
	in a pass for updating~$Z_p$
	is~$\bigoof{n^3}$.
	Since the number of passes is at most~$n$,
	the total number of visits to each cell is~$\bigoof{n^4}$.
	Moreover,
	this operation can be implemented by using~$\bigoof{n^3}$ states
	because
	we just need to remember the values of~$r$, $s$, and $r'$,
	but not the number of the pass:
	it is sufficient to use a bit
	to remember if during the last pass at least one state has been added to~$Z_p$.
	This is because if no states are added during a pass,
	then no states will be added executing further iterations.

	Once~$Z_p$ has been computed,
	for each state~$r$ corresponding to a marked cell in the temporary table,
	and each state~$q$ such that ${(q,\rmove)\in\delta(r,\sigma)}$,
	\mTM[\sigma] adds
	the pair~$(p,q)$
	to the table~$\tau_{z\sigma}$
	represented on the second track of~$\overline v$.
	This requires to visit~$\bigoof n$ times each tape cell
	and can be performed using only a quadratic number of states in~$n$.

	By repeating this for each state~$p$,
	we manage to update the table from~$\tau_z$ to~$\tau_{z\sigma}$.
	Finally,
	we can update the table~$\gamma_z$.
	As observed before,
	it is sufficient to consider for each state~$q$,
	whether~$(p,q)\in\tau_{z\sigma}$ for some~$p\in\gamma_z$.
	This last step requires only a quadratic number of states
	and a linear number of visits to each cell.
	Combining the above-described subroutines,
	and taking into account that the state component
	which stores the head position has size~$\bigoof{n^2}$,
	we obtain a \dhm with~$\bigoof{n^6}$ states
	and~$\bigoof 1$ working symbols,
	whose number of visits to each tape cell
	is in~$\bigoof{n^5}$.
\end{proof}

We are now ready to state our main simulation.
\begin{theorem}
	\label{thm:2nfa->wrdtm}
	Every $n$-state \twnfa
	can be transformed into an equivalent
	halting \wrdtm
	of size polynomial in~$n$.
\end{theorem}
\begin{proof}
	Let~$\mA=\struct{Q,\ialphab,\delta,q_0,F}$
	be a \twnfa.
	We build a deterministic Turing machine
	that mimics the simulation of~\mA
	by a \owdfa:
	after reading any prefix~$z$ of an input~$w$,  
	the machine stores the tables~$\gamma_{\lend z}$ and~$\tau_{\lend z}$
	and finally checks the existence of a final state
	in~$\gamma_{\lend w\rend}$.
	The tables are stored on a suitable tape track
	and updated each time
	a further input symbol is read,
	using the method presented in \cref{lem:update-tables}.
	This can be achieved by switching between two tape tracks
	at each update of the tables.
	However, as the number of updates
	is linear in the length of the input,
	storing and updating the tables
	on a fixed part of length~$n+n^2$ of the tape
	would lead to a non-weight-reducing Turing machine.
	To handle this issue,
	for each prefix~$z$ of~$w$,
	we store the tables~$\gamma_{\lend z}$ and~$\tau_{\lend z}$
	on the~$n+n^2$ cells that precede the last position of~$z$.
	(%
		Remember that,
		in a \wrdtm,
		some initially-blank cells
		to the left of the initial segment
		are available.%
	)
	Thus,
	at each update of the tables made according to \cref{lem:update-tables},
	the tables are shifted one cell to the right.
	Hence,
	since a fixed cell may occur in~$n+n^2$ successive stored tables,
	the number of visits to each cell is in~$\bigoof{n^7}$,
	and the number of states is in~$\bigoof{n^6}$.
	We thus obtain a halting weight-reducing Turing machine equivalent to~\mA
	whose size is polynomial in the size of~\mA by \cref{lemma:constantly-many-visits}.
	Furthermore,
	the machine uses only~$n+n^2$ initially-blank cells,
	that are all to the left of the initial segment.
\end{proof}
\end{document}

\section{Simulating Two-way Automata by Hennie Machines}
\label{sec:2nfa->dhm}

In~\cref{sec:2nfa->wrotm}
we provided a polynomial size conversion
from \twnfa\s to \wrdtm\s.
The resulting machines use further tape cells,
besides the initial segment.
In this section we study
how to make such a simulation when the use of such extra space is not allowed,
namely
when we want to obtain a deterministic Hennie machine.
We show that a polynomial conversion still exists,
but we are not able to guarantee that the resulting machine is weight reducing.
Actually
this issue is related to ``short'' inputs,
namely to strings of length less than~$n^2$,
where~$n$ is the number of states of the given~\twnfa.
For such inputs we do not have enough tape space
to perform the simulation in~\cref{thm:2nfa->wrdtm}.
We will deal with them,
by using a different technique.

Let us start by considering acceptance of ``long'' inputs.
\begin{theorem}
	\label{thm:2nfa:space-n}
	\label{thm:2nfa->wrdhm long inputs}
	For each $n$-state \twnfa~\mA,
	there exists a \wrdhm~\mHM
	of size polynomial in~$n$
	which agrees with~\mA on strings of length at least~$n^2$.
\end{theorem}
\begin{proof}
	The technique used in the proof of \cref{thm:2nfa->wrdtm} 
	can be exploited,
	with slight modifications.
	Indeed,
	when recovering the tables corresponding to the ``short'' prefixes~$z$'s of the tape content,
	the \wrdtm machine resulting from the above construction
	uses up to~$n+n^2$ initially-blank cells to the left of the initial segment,
	that are not any longer available with a \wrdhm.
	By folding~$n$ of these cells
	on an additional track,
	as in the proof of \cref{lem:wrTm->wrHm},
	we can reduce this space amount to~$n^2$ cells by a polynomial size increasing.
	Hence, applying \cref{lem:wrTm->wrHm},
	we can obtain a \wrdhm which agrees with the original machine on inputs of length at least~$n^2$.
\end{proof}

In the case of \defd{unary} \twnfa\s,
namely working on a single-letter input alphabet,
the number of short inputs
that are not handled by \cref{thm:2nfa:space-n} is~$n^2$.
They can be managed in a read-only preliminary 
phase which uses~$\bigoof{n^2}$ states.
\begin{theorem}
	\label{cor:u2nfa:space-n}
	\label{cor:u2nfa->wrdhm}
	Every $n$-state unary \twnfa is equivalent to a \wrdhm
	of size polynomial in~$n$.
\end{theorem}
\begin{proof}
	We simulate a given unary \twnfa~\mA as follows.
	Let~$X$ be the finite set
	$\set{i<n^2}[i=\length w\text{ for some }w\in L]$.
	First, the head of our simulating \wrdhm is moved rightward
	to test whether the input is shorter than~$n^2$,
	using a counter from~$0$ to~$n^2-1$.
	In this case,
	the machine accepts \ifof the counter value belongs to~$X$.
	Otherwise,
	the head is moved back to the left endmarker
	and the simulation from \cref{thm:2nfa->wrdhm long inputs} is performed.
	With respect to the \wrdtm obtained from \cref{thm:2nfa->wrdhm long inputs},
	our device uses~$\bigoof{n^2}$ additional states and $\bigoof{1}$ extra working symbols.
	Hence, our construction yields a \wrdhm equivalent to~\mA
	of size polynomial in~$n$.
\end{proof}

In the nonunary case,
since the number of short strings is exponential in~$n$,
we cannot apply the same technique
as in~\cref{cor:u2nfa:space-n}.
However,
we are able to obtain a polynomial size Hennie machine
(not necessarily weight reducing),
using a different technique,
which is based on the analysis of the computation graph
of the simulated \twnfa.
\begin{theorem}
	\label{thm:2nfa:space-logn}
	\label{thm:2nfa->dhm}
	Each $n$-state \twnfa is equivalent to a \dhm of size polynomial in~$n$.
\end{theorem}
\begin{proof}
	Let~$\mA=\struct{Q,\alphab,\delta,q_0,F}$
	be a \twnfa,
	with~$\card Q=n$.
	Without loss of generality
	we suppose~$F=\set{q_f}$.
	Let~$w\in\alphab^{\ast}$ be an input word,
	with~${m=\length w}$.
	We distinguish three cases,
	depending on~$m$.
	Observe that the simulating~\dhm~\mHM can decide the case
	by performing a reading traversal of the input
	using a polynomial number of states.

	If~$m\geq n^2$,
	then~\mHM simulates~\mA as in \cref{thm:2nfa:space-n}.
	
	If~$m\leq\log n$,
	then~\mHM simulates a \owdfa
	with a polynomial number of states in~$n$
	(and in the number of input symbol which is assumed to be a fixed constant),
	which agrees with~\mA on all strings of length at most~$\log n$.

	Finally, if~$\log n<m<n^2$, then
	\mHM checks whether there is an accepting computation
	of~\mA on~$w$
	by analyzing the \emph{computation graph}~$G=\struct{V,E}$ of~\mA on~$w$,
	defined as follows.
	The set of vertices of~$G$ is~$V=Q\times\set{0,\ldots,m+1}\cup\set{(q_f,m+2)}$,
	where the pair~$(q,i)\in V$ corresponds to the configuration on input~$w$ in which~\mA is in the state~$q$
	while scanning the $i$-th symbol of the input tape.
	The edges in~$E$ represent single moves,
	\ie, there exists an edge from~$(q,i)$ to~$(p,j)$
	if and only if
	$(p,j-i)\in\delta(q,\ith{\extword w})$,
	where~$\extword w=\lend w\rend$.
	
	The simulating Hennie machine~\mHM should check
	the existence of
	a computation of~\mA
	starting from the initial state~$q_0$
	with the head on the left endmarker (\ie, at position~$0$)
	and
	ending in the unique final state~$q_f$
	after violating the right endmarker (\ie, at position~$m+2$).
	This is equivalent to
	check the existence of a path
	from the node~$(q_0,0)$
	to the node~$(q_f,m+2)$
	in~$G$.
	Let~$K=n(m+2)+1$ be the number of nodes in~$G$.
	If such a path exists,
	then there should exist one of length at most~$K$.
	Hence,
	checking the existence of an accepting computation
	reduces to checking the existence of a path of length at most~$K$ in~$G$.
	The recursive function \reachable
	is used to perform this checking
	by calling \reachable{$q_0,0,q_f,m+2,K$}.

	\begin{function}[bt]
		\caption{%
			reachable($p,i,q,j,T$): boolean\newline
			Checks the existence of a path
			from~$(p,i)$ to~$(q,j)$
			of length less than or equal to~$T$	
			in the graph of the configurations of a given~\twnfa on input~$w=w_1\cdots w_m$}
		\lIf{$(p,i)=(q,j)$}{\Return \KwTrue}\label{proc:reachable:len0}
		\lIf{$T = 0$}{\Return \KwFalse}\label{proc:reachable:unfeasable}
		\uIf{$T=1$}{\label{proc:reachable:direct_edge:start}
			\lIf{$(q,j-i)\in\delta(p,\ith{\extword w})$}{\Return \KwTrue}
		}\label{proc:reachable:direct_edge:end}
		\Else{
			\ForEach{$r,\ell \in Q\times\set{0,\ldots,m+1}$}{
				\label{proc:reachable:subpath:start}
				\If{\reachable{$p,i,r,\ell,\floor{T/2}$}}{
					\lIf{\reachable{$r,\ell,q,j,\ceil{T/2}$}}{\Return\KwTrue}
				}
				\label{proc:reachable:subpath:end}
			}
		}
		\Return \KwFalse
		\label{proc:reachable:fail}
	\end{function}
	Let us describe how a call of \reachable{$p,i,q,j,T$} works.
	The function has to check
	if there exists a computation from~$(p,i)$ to~$(q,j)$
	of length at most~$T$.
	This is done by	using a \emph{divide-and-conquer} technique
	as in the famous proof of the Savitch's Theorem~\cite{Sav70}.
	If~$(p,i)=(q,j)$,
	\ie,
	there is a path of length~$0$
	from~$(p,i)$ to~$(q,j)$,
	then the function returns \KwTrue
	independently of~$T$
	(Line~\ref{proc:reachable:len0}).
	Otherwise,
	if~$T=0$ but~$(p,i)\neq(q,j)$,
	then the function returns \KwFalse
	(Line~\ref{proc:reachable:unfeasable}),
	while,
	if~$T=1$,
	the function returns \KwTrue if there is a suitable edge in~$G$
	(%
		Line~
		\ref{proc:reachable:direct_edge:end}%
	).
	In order to verify that,
	\mHM saves~$(q,i)$ and~$(p,j)$ in its internal state
	and then,
	if the distance between~$i$ and~$j$ is~$1$,
	it moves its head to position~$i$,
	reads the symbol~$\ith{\extword w}$,
	and checks~$(q,j-i)\in\delta(p,\ith{\extword w})$.
	Notice that this read-only process
	uses only a number of states polynomial in~$n$.

	In the recursive case,
	for checking if there exists a path in the graph
	from~$(p,i)$ to~$(q,j)$
	of length at most~$T>1$,
	\mHM verifies whether there exists a node
	$(r,\ell)\in Q\times\set{0,\ldots,m+1}$
	such that
	there is a path from~$(p,i)$ to~$(r,\ell)$
	of length at most~$\floor*{\frac T 2}$
	and a path from~$(r,\ell)$ to~$(q,j)$
	of length at most~$\ceil*{\frac T 2}$
	(%
		Lines~\ref{proc:reachable:subpath:start}
		to~\ref{proc:reachable:subpath:end}%
	).
	This is done by trying all possible nodes~$(r,\ell)$
	until finding one satisfying
	these conditions.
	If it does not exist,
	then the procedure returns \KwFalse
	(Line~\ref{proc:reachable:fail}).
	
	Recursive calls to the function \reachable
	can be naturally saved on a pushdown store.
	More precisely,
	at the beginning of the simulation
	the store is empty.
	When a call to \reachable{$p,i,q,j,T$}
	is performed,
	the activation record,
	consisting of the parameters~$p$, $i$, $q$, $j$, and~$T$,
	is pushed
	on the top of the pushdown.
	Similarly,
	when \reachable returns,
	the activation record of the last call
	is popped off.
	The function \reachable uses
	seven variables,
	five of them being arguments
	saved on the pushdown store,
	and two
	being the local variables~$r$ and~$\ell$.
	As these two local variables
	are arguments of inner recursive calls,
	their values can be recovered
	when popping off the inner activation record
	(after the corresponding call has returned).
	Hence, the state components saving~$r$ and~$\ell$
	are freed at each recursive call.
	Therefore,
	all the checks performed by \reachable
	can be done with a number of states
	that is
	polynomial in~$n$,
	and using the pushdown alphabet~%
	$
		\left(Q \times \set{0,\ldots,n^2}\right)^2
		\times
		\set{0,\ldots,\ceil{\log K}}
	$
	of size polynomial in~$n$.

	Finally,
	notice that the maximum recursion depth 
	is%
	~$\ceil*{\log K}=\bigoof{\log n}$.
	The stack of recursion calls can be stored in a separated track on~$\log n$ tape cells,
	by using standard space compression techniques,
	that only induce a polynomial increase of the working alphabet.
	Since the input length~$m$ is larger than~$\log n$,
	\mHM has enough space in its initial segment.
	The number of visits to each cell is super-polynomial in~$n$,
	hence
	the machine is not weight-reducing.
	However,
	because the input lengths are bounded by~$n^2$,
	the number of visits to each cell is bounded by a number which only depends on~$n$.

	Considering also how the machine works on inputs of length at least~$n^2$,
	this allows us to conclude that
	the working time of the whole machine \mHM is linear in the input length.
\end{proof}

It is natural to ask if~\cref{thm:2nfa:space-logn} can be improved in order to obtain from a given \twnfa~\mA an equivalent \wrdhm of polynomial size. 
In the light of~\cref{thm:2nfa:space-n}, to do that it will be enough to obtain a \wrdhm of polynomial size which agrees with the \twnfa on ``short'' inputs.
With this respect, we point out that the problem of Sakoda and Sipser seems to be hard even when restricted to strings of length polynomial in the number of
states of~\mA~\cite{Kap14}.

\end{document}

\section{The One-Way Case}
\label{sec:1nfa->dhm}

In this section we restrict our attention to one-way automata simulations.
A natural question is to ask
if in the case of \ownfa\s
results stronger than those presented in \cref{sec:2nfa->dhm} can be achieved.
A simulation of \ownfa\s by \wrdhm\s
was studied in~\cite[Theorem~11]{Pru14},
claiming that each $n$-state \ownfa~$\mA$
has an equivalent \wrdhm~$\mHM$
of size polynomial in $n$.
Unfortunately,
the presented proof is incorrect
as it casts the problem of \mA acceptance
as the problem of reachability in an undirected computation graph.
Existence of a path connecting the initial and an accepting configuration
in such a graph
does not guarantee the existence of an accepting computation of $\mA$
since the path can include ``back'' edges.

By revising~\cite[Theorem~11]{Pru14}, 
we prove two weakened variants of this result.
In the first one,
the simulation holds only for long enough inputs.
The improvement with respect to \cref{thm:2nfa->wrdhm long inputs}
is that, in this case, short inputs are the strings
of length less than~$n$ rather than~$n^2$.
In the second variant,
we show that each \ownfa can be simulated
by a deterministic Hennie machine
which, however, is not weight-reducing.

Let us start by presenting the weight-reducing simulation for ``long'' inputs.
\begin{theorem}
	\label{thm:1nfa:space-n}
	\label{prop:1nfa->wrdhm long inputs}
	For each $n$-state \ownfa~\mA,
	there exists a \wrdhm~\mHM
	of size polynomial in~$n$
	which agrees with~\mA
	on strings of length at least~$n$.
\end{theorem}
\begin{proof}
	In order to obtain~\mHM,
	we build a \dtm~\mTM equivalent to~\mA
	with the following properties:
	\begin{enumerate}[label={(P\arabic*)},ref={(P\arabic*)},nosep]
		\item\label{p:size}
			the size of~\mTM is bounded by a polynomial in the size of~\mA;
		\item\label{p:visits}
			the number of visits to each cell in any computation
			is bounded by a polynomial in~$n$;
		\item\label{p:extra}
			only~$n$ tape cells beside the initial segment are used,
			in any computation.
	\end{enumerate}
	Then, using \ref{p:visits} and \cref{lemma:constantly-many-visits},
	\mTM can be turned into an equivalent \wrdtm
	preserving \ref{p:size} and~\ref{p:extra}.
	Finally, using \ref{p:extra} and \cref{lem:wrTm->wrHm},
	the resulting \wrdtm is converted to the wanted \wrdhm~\mHM
	whose size is polynomial in the size of~\mA.

	We build~\mTM by adapting the classical powerset construction
	for converting \ownfa\s into \owdfa\s.
	In this standard construction,
	the simulating \owdfa scans the input word,
	while keeping updated, at each step,
	the subset of states
	that can be reached by~\mA
	from the initial configuration while reading input symbols.
	Our construction uses this technique but,
	instead of storing the subset in the finite control of the simulating machine,
	\mTM stores the successive subsets on the tape,
	exploiting the writing capability of~\wrdtm\s
	in order to avoid an exponential size blowup.

	Let~$\mA=\struct{Q,\alphab,\delta,q_0,F}$
	with~$Q=\set{q_0,q_1,\ldots,q_{n-1}}$.
	Following the notations of \cref{sec:2nfa->wrotm},
	for any input prefix~$z$,
	we consider the set~$\gamma_z\subseteq Q$
	of states that are reachable from the initial configuration
	after having read~$z$.
	In particular, $\gamma_\emptyword$ is the singleton~$\set{q_0}$
	and an input~$w$ is accepted by~\mA \ifof~$\gamma_w\cap F$ is nonempty.
	In order to store any subset~$\gamma_z$ of~$Q$ on the tape,
	we represent it as a word~$u\in\set{0,1}^n$,
	in which~$u_i$ is~$1$ \ifof~$q_i\in\gamma_z$.
	Ranging over input prefixes,
	\mTM iteratively computes~$\gamma_{z\sigma}$ from~$\gamma_z$,
	where~$\sigma\in\ialphab$ is the~$(\length{z}+1)$-th input symbol.
	To avoid erasing~$\gamma_z$
	when writing down~$\gamma_{z\sigma}$,
	\mTM uses two distinct tape tracks:
	$\gamma_z$ is read from the first track
	and~$\gamma_{z\sigma}$ is written on the second one
	(shifted by one cell to the right with respect to~$\gamma_z$).
	The role of the two tracks is inverted after each such update.

	At the beginning of the computation,
	\mTM writes~$\gamma_\emptyword$
	on the~$n$ cells preceding the initial segment
	(on one track, suppose the first one for ease of exposition).
	During the simulation,
	when \mTM reaches a position for the first time,
	the~$n$ tape cells preceding that position
	will contain the encoding of~$\gamma_z$ (on the first track)
	where~$z$ is the input prefix read so far.
	From such a point,
	\mTM is able to compute~$\gamma_{z\sigma}$,
	working on the tape portion of length~$n+1$
	that contains both~$\gamma_z$ and the current position containing~$\sigma$,
	and to which we refer as the \emph{current working segment}.
	Indeed, for each~$q_i\in Q$,
	we have~$q_i\in\gamma_{z\sigma}$
	\ifof $q_i\in\delta(q_j,\sigma)$ for some~$q_j\in\gamma_z$.
	Since~\mTM can test for each~$q_i,q_j\in Q$
	whether, on the one hand,~$q_j$ belongs to~$\gamma_z$
	(by visiting the~$j$-th cell of the current working segment
	which stores the~$j$-th bit of~$\gamma_z$ on the first track),
	and, on the other hand,~$q_i\in\delta(q_j,\sigma)$
	(by inspecting the transition table of~\mA),
	it can write~$\gamma_{z\sigma}$ on the tape.
	More precisely,
	when the two above conditions are satisfied,
	\mTM writes~$1$ on the~$(i+1)$-th cell of the current working segment (on the second track),
	for storing the~$i$-th bit of~$\gamma_{z\sigma}$.
	Notice that~$\gamma_{z\sigma}$ is shifted one cell to the right with respect to~$\gamma_z$,
	so that its last bit is written on the rightmost cell of the current working segment
	(which initially contains~$\sigma$).
	Furthermore, during the update procedure,
	each cell of the current working segment is visited~$\bigoof{n}$ times.
	Finally, when reaching the first cell to the right of the initial segment
	(containing the blank symbol),
	\mTM accepts \ifof~$\gamma_w$ contains a final state.
	By slightly modifying the above update procedure,
	we may assume that this information has been prepared in the finite control.
	Indeed, when computing~$\gamma_{z\sigma}$,
	\mTM can store in a state component
	whether~$\gamma_{z\sigma}$ intersects~$F$.

	The described approach implies the use of~$n$ extra tape cells
	to the left of the initial segment,
	on which \mTM has initially written~$\gamma_\emptyword$.
	No extra space is necessary,
	thus~\ref{p:extra} is satisfied.
	Moreover,
	the simulating machine uses~$\bigoof{1}$ working symbols,
	and~$\bigoof{\card{\ialphab}\cdot n^3}$ states
	for simultaneously storing the variables~$\sigma$, $q_i$, $q_j$,
	and the head position relative to the current working segment.
	This implies \ref{p:size}.
	Next, each tape cell is used for storing at most~$n$ successive subsets of states.
	Hence, its total number of visits is in~$\bigoof{n^2}$,
	yielding \ref{p:visits}.
%
\end{proof}

By \cref{thm:2nfa:space-logn},
every $n$-state \twnfa can be transformed into an equivalent \dhm of size polynomial in~$n$.
So, using the same result,
we can also transform~\ownfa\s.
For this particular case,
we now present a more direct simulation
which uses a technique similar to that of proof of \cref{thm:2nfa:space-logn}.
Let~$\mA=\struct{Q,\alphab,\delta,q_0,F}$ be an~$n$-state \ownfa.
At the cost of one extra state, we suppose~$F=\set{q_f}$.
Let~$w \in \Sigma^*$ be an input to~\mA.
Distinguish two cases by $m=\length w$.

If~$m\leq n$,
the machine~\mHM checks whether there is an accepting computation of~\mA
on~$w$ by calling $\reachOneWay(q_0,0,q_f,m)$,
a slightly modified version of the function $\reachable$ presented in the proof of \cref{thm:2nfa:space-logn},
adapted to deal with \ownfa\s.
\begin{function}[ht]
	\caption{reachableOneWay($p,i,q,j$): boolean\newline
		Checks the existence of a path
		from~$(p,i)$ to~$(q,j)$
		of length $j-i$	
		in the graph of the configurations of a given~\ownfa on input~$w=w_1\cdots w_m$}
	\lIf{$(p,i)=(q,j)$}{\Return \KwTrue}
	\lIf{$j-i=1$ \KwAnd $q\in \delta(p,w_j)$}{\Return \KwTrue}
	\If{$j-i>1$}
	{
		\ForEach{$r\in Q$}{
			\If{$\reachOneWay(p,i,r,\floor{(i+j)/2})$}{
				\lIf{$\reachOneWay(r,\ceil{(i+j)/2},q,j)$}{\Return\KwTrue}
			}
		}
	}
	\Return \KwFalse
\end{function}

The recursion depth is $\bigoof{\log m}$.
At each level of the recursion,
the function needs to store states~$p$, $q$, $r$ and indices~$i$, $j$.
This can be done in a tape cell
if~$\bigoof{n^5}$ working tape symbols are provided for this purpose.
The function runs in $n^{\bigoof{\log m}}$ time. 

If~$m>n$,
the Hennie machine~\mHM executes the computation
described in the proof of \cref{thm:1nfa:space-n}.

In conclusion,
the constructed machine~\mHM
has the number of states and working tape symbols polynomial in~$n$,
hence it is of size polynomial in~$n$.
Note that the number of transitions
performed by~\mHM over any tape field
is~$n^{\bigoof{\log n}}$.

\end{document}

\section{Conclusion}
\label{sec:conclusion}
Sakoda and Sipser raised the question of the cost of the elimination of nondeterminism from finite automata
exploiting the possibility of moving the head in both directions.
Though they conjectured that this cost is exponential,
we proved
that a determinization of polynomial size cost is possible
for some simulating machines which are deterministic and have some extra capabilities than \twdfa\s.
The extensions of \twdfa\s we considered are some special cases of one-tape Turing machines which are not more powerful than \twdfa\s,
namely they recognize the class of regular languages.

In \cref{thm:2nfa->dhm}
we showed such a polynomial determinization using Hennie machines.
However, this result is not fully satisfying
because of known drawbacks of \dhm\s:
on the one hand, it is not decidable
whether a Turing machine is actually a Hennie machine;
on the other hand, no recursive function bounds
the size blowup of the conversion of \dhm\s into \owdfa\s.
These drawbacks are avoided when moving to weight-reducing devices~\cite{GPPP21a}.

We do not know if the simulations of~\cref{thm:2nfa->dhm} can be changed
in such a way that the resulting \dhm is always weight-reducing
while keeping the size cost polynomial,
and we leave it as an open problem.
However, we solve this question in some particular cases.
Indeed, when equivalence is required only over long enough inputs,
or when the input alphabet is unary,
polynomial size determinizations using \wrdhm\s are possible
(\cref{thm:2nfa->wrdhm long inputs,cor:u2nfa->wrdhm,prop:1nfa->wrdhm long inputs}).
Moreover,
our main result states
that it is always possible to eliminate nondeterminism from finite automata
by using weight-reducing linear-time Turing machines
as long as they are not required to be end-marked (\cref{thm:2nfa->wrdtm}),
namely when the simulating machine is allowed to use extra tape beside the part that initially contains the input.

\bibliography{wrhm}
\end{document}